\documentclass[11pt]{article} 

\usepackage[utf8]{inputenc} 
\usepackage{geometry} 
\usepackage{graphicx} 
\usepackage{color}

\definecolor{Red}{rgb}{1,0,0}
\definecolor{Blue}{rgb}{0,0,1}
\definecolor{Olive}{rgb}{0.41,0.55,0.13}
\definecolor{Green}{rgb}{0,1,0}
\definecolor{MGreen}{rgb}{0,0.8,0}
\definecolor{DGreen}{rgb}{0,0.55,0}
\definecolor{Yellow}{rgb}{1,1,0}
\definecolor{Cyan}{rgb}{0,1,1}
\definecolor{Magenta}{rgb}{1,0,1}
\definecolor{Orange}{rgb}{1,.5,0}
\definecolor{Violet}{rgb}{.5,0,.5}
\definecolor{Purple}{rgb}{.75,0,.25}
\definecolor{Brown}{rgb}{.75,.5,.25}
\definecolor{Grey}{rgb}{.5,.5,.5}

\usepackage{booktabs} 
\usepackage{array} 
\usepackage{paralist} 
\usepackage{verbatim} 
\usepackage{subfigure} 
\usepackage{amsmath,amssymb,amsthm}
\usepackage{psfrag}

\usepackage{fancyhdr} 
\usepackage{sectsty}
\allsectionsfont{\sffamily\mdseries\upshape} 

\theoremstyle{plain}

\newtheorem{theorem}{Theorem}

\newtheorem{claim}{Claim}
\newtheorem{lemma}{Lemma}

\theoremstyle{remark}

\newtheorem{remark}{Remark}

\theoremstyle{definition}
\newtheorem{definition}{Definition}

\newcommand{\p}{{\rm P}}

\def\cP{{\cal P}}

\def\cW{{\cal W}}
\def\cX{{\cal X}}
\def\cY{{\cal Y}}

\title{Capacity regions of two new classes of 2-receiver broadcast channels}
\author{Chandra Nair\\
Department of Information Engineering\\
Chinese University of Hong Kong \\
Sha Tin, N.T., Hong Kong\\
Email: chandra@ie.cuhk.edu.hk}

\date{}

\begin{document}
\maketitle


\begin{abstract}
Motivated by a simple broadcast channel, we generalize the notions of a {\em less noisy} receiver and a {\em more capable} receiver to an {\em essentially less noisy receiver} and an {\em essentially more capable} receiver respectively. We establish the capacity regions of these classes by borrowing on existing techniques to obtain the characterization of the capacity region for certain new and interesting classes of broadcast channels. We also establish the relationships between the  new classes and the existing classes.
\end{abstract}



\section{Introduction}
\label{se:intro}

This paper is motivated directly by a simple broadcast channel setting, posed by Andrea Montanari(see Figure \ref{fig:bc}), consisting of a BSC(p) and BEC(e). Clearly if $e \leq 2p$, then the channel is {\em degraded}\cite{cov72} and the capacity\cite{gal74,ber73} is given by the union of rate pairs $(R_1,R_2)$
satisfying
\begin{align*}
R_1 &\leq I(U;Y_1) \\
R_1 + R_2 &\leq I(U;Y_1) + I(X;Y_2|U)
\end{align*}
over all $(U,X)$ such that $U \to X \to (Y_1,Y_2)$ form a Markov chain.

\begin{figure}[ht]
\begin{center}
\begin{psfrags}
\psfrag{a}[c]{$Y_1$}
\psfrag{b}[c]{$Y_2$}
\psfrag{c}[c]{$X$}
\psfrag{d}[c]{$p$}
\psfrag{e}[c]{$e$}
\includegraphics[width=0.4\linewidth,angle=0]{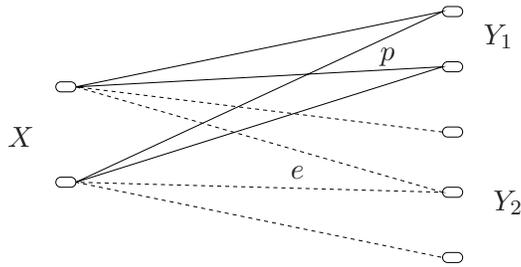}
\end{psfrags}
\caption{A broadcast channel consisting of a BSC(p) and BEC(e)}
\label{fig:bc}
\end{center}
\end{figure}

If $1 - H(p) \leq 1 - e$ then $Y_2$ would be a {\em more capable}\cite{kom75} receiver than $Y_1$ (see Parts 1,2 of Claim \ref{cl:easy} in the Appendix) and in this case capacity\cite{elg79} is given by the union of rate pairs satisfying 
\begin{align*}
R_1 &\leq I(U;Y_1) \\
R_1 + R_2 &\leq I(U;Y_1) + I(X;Y_2|U) \\
R_1 + R_2 &\leq I(X;Y_1)
\end{align*}
over all $(U,X)$ such that $U \to X \to (Y_1,Y_2)$ form a Markov chain. Therefore, the interesting case for determining the capacity region occurs when $1 - H(p) > 1 - e$.

Hence we restrict ourselves to the case when $1 - H(p) > 1 - e$, i.e. when the channel $X \to Y_1$ has a higher capacity than the channel $X \to Y_2$ . The Figure \ref{fig:btwo} plots $I(X;Y_1) - I(X;Y_2)$ for the case $p=0.1, e=0.5$. It is clear that neither is more capable than the other. In particular this setting does not fall into any class of broadcast channels for which the capacity region has been characterized. We address this regime and establish the capacity region. 

In fact we  establish the capacity region of a whole new class of broadcast channels (motivated of course by this example) that contains this broadcast channel, under the regime $1 - H(p) > 1 - e$, as a special case. 

\begin{figure}[ht]
\begin{center}
\begin{psfrags}
\psfrag{a}[c]{$I(X;Y_1) - I(X;Y_2)$}
\psfrag{b}[c]{$\p(X=0)$}
\includegraphics[width=0.4\linewidth,angle=0]{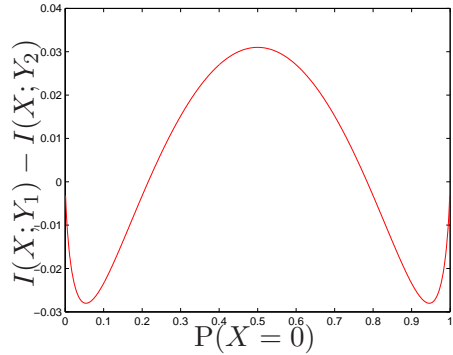}
\end{psfrags}
\caption{The function $I(X;Y_1) - I(X;Y_2)$ for BSC(0.1) and BEC(0.5)}
\label{fig:btwo}
\end{center}
\end{figure}

\subsection{Observation}
\label{subsec:obser}

The common theme between degraded, less noisy, and more capable channels is the existence of a dominant receiver who manages to decode the private messages for both the users. Let us recall the definition of the {\em less noisy}\cite{kom75} receiver. One requires  $ I(U;Y_2) \leq I(U;Y_1) $ 
to hold true for {\em every} $p(u,x)$ to classify receiver $Y_1$ as a less noisy receiver than $Y_2$. In this paper we remove the need for $ I(U;Y_2) \leq I(U;Y_1) $ 
holding true for {\em every} $p(u,x)$ and replace it by  $ I(U;Y_2) \leq I(U;Y_1) $ 
holding true for {\em a sufficiently large} class of distributions $p(u,x)$. We then show that for this relaxed definition of a less noisy receiver as well, the capacity regions can be obtained.

A similar development has been done for a more capable receiver as well.

\begin{remark}
The key contribution of this paper is the identification of capacity regions for  interesting classes of broadcast channels by proving that these channels belong to a slightly tweaked definition of a less noisy or a more capable receiver. Secondly, the tweaking of the definitions is done in such a way that the existing techniques are sufficient to establish the capacity regions.
\end{remark}

The organization of the paper is as follows. In section \ref{se:defn} we make the formal definitions and set up the required notation. In Section \ref{se:cap} we establish the capacity region of a class of two-receiver broadcast channels that has one receiver who is essentially less noisy when compared to the other.  In Section \ref{se:essln}  we identify certain interesting classes of channels that are neither less noisy nor more capable but are essentially less noisy. In Section \ref{se:mcap} we establish the capacity region of a class of two-receiver broadcast channels that has one receiver who is essentially more capable than the other. Finally in Section \ref{se:inc} we establish the various inclusions among these classes.

\section{Definitions and Notation}
\label{se:defn}

In \cite{cov72}, Cover introduced the notion of a broadcast channel through which one sender transmits information to two or more receivers. For the purpose of this paper we focus our attention on broadcast channels with precisely two receivers.

A {\em broadcast channel} (BC) consists of an
input alphabet $\mathcal{X}$ and output alphabets $\mathcal{Y}_1$
and $\mathcal{Y}_2$ and a probability transition function
$p(y_1,y_2|x)$. A $((2^{nR_1}, 2^{nR_2}),n)$ code for a broadcast
channel consists of an encoder
\[ X^n : \cW_1 \times \cW_2 \rightarrow \mathcal{X}^n, \]
and two decoders
\[ \hat{W}_1 : \mathcal{Y}_1^n \rightarrow \cW_1 \]
\[ \hat{W}_2 : \mathcal{Y}_2^n \rightarrow \cW_2, \]
where $\cW_1 = \{1,2,...,2^{nR_1}\}, \cW_2 = \{1,2,...,2^{nR_2}\}$.

The probability of error $\p_e^{(n)}$ is defined to be the
probability that the decoded message is not equal to the transmitted
message, i.e., 

\[ \p_e^{(n)} = \p \left(\{\hat{W}_1(Y_1^n) \neq W_1 \}\cup \{
\hat{W}_2(Y_2^n) \neq W_2 \} \right)
\]
where the message pair $(W_1,W_2)$  is assumed to be uniformly distributed over
$\cW_1 \times \cW_2$. 

A rate pair $(R_1,R_2)$ is said to be {\em achievable } for the
broadcast channel if there exists a sequence of $((2^{nR_1},
2^{nR_2}),n)$ codes with $\p_e^{(n)} \rightarrow 0$. The {\em
capacity region} of the broadcast channel with is the closure of the set of achievable rates.
{\em 
 The capacity region of the two user discrete memoryless channel is unknown.}

The capacity region is known for lots of special cases such as degraded, less noisy, more capable, deterministic, semi-deterministic, etc. - see  \cite{cov98} and the references therein. In this paper we establish the capacity region for two more classes of broadcast channels, one where one receiver is {\em essentially less noisy} compared to the other receiver; and the other where one receiver is {\em essentially more capable} than the other receiver. 

A channel $X \to Y_2$ is said to be a {\em degraded} version of the channel $X \to Y_1$ if $X \to Y_1 \to Y_2'$ is a Markov chain and the pair $(X,Y_2')$ is identically distributed as the pair $(X,Y_2)$. A receiver $Y_1$ is said to be {\em less noisy}\cite{kom75} compared to $Y_2$ if
$$ I(U;Y_2) \leq I(U;Y_1) $$
for all $p(u,x)$ such that $U \to X \to (Y_1,Y_2)$ forms a Markov chain. Finally, a receiver $Y_1$ is said to be {\em more capable}\cite{kom75} compared to $Y_2$ if
$$ I(X;Y_2) \leq I(X;Y_1) $$
for all $p(x)$.

\medskip

\begin{definition}
A class of distributions $\mathcal{P} = \{p(x)\}$ on the input alphabet $\mathcal{X}$ is said to be a {\em sufficient class} of distributions for a 2-receiver broadcast channel if the following holds:
Given any triple of random variables $(U,V,X)$ distributed\footnote{In all cases we assume that the tuple $(U,V,X,Y_1,Y_2)$ satisfies $(U,V) \to X \to (Y_1,Y_2)$ forms a Markov chain. In a discrete memoryless broadcast channel with no feedback this assumption is "automatically" satisfied. However it is necessary to state it explicitly to prevent choices like $U=Y_1$(except when $X \to Y_1$ is deterministic) and other strange choices.}  according to $p(u,v,x)$, there exists a distribution $q(u,v,x)$ that satisfies
\begin{align}
& q(x) \in \mathcal{P}, \nonumber \\
&I(U;Y_i)_p \leq I(U;Y_i)_q,~ i=1,2, \nonumber \\
&I(V;Y_i)_p \leq I(V;Y_i)_q, ~i=1,2, \nonumber \\
&I(X;Y_i|U)_p \leq I(X;Y_i|U)_q, ~i=1,2, \label{eq:suff}\\
&I(X;Y_i|V)_p \leq I(X;Y_i|V)_q, ~i=1,2, \nonumber \\
&I(X;Y_i)_p \leq I(X;Y_i)_q, ~i=1,2,\nonumber 
\end{align}
\end{definition} 

The notation $I(U;Y_1)_p$ denotes the mutual information between $U$ and $Y_1$ when the input is generated using $p(u,v,x)$.
\medskip

\begin{definition}
A receiver $Y_1$ is {\em essentially less noisy} compared to receiver $Y_2$ if there exists a sufficient class of distributions $\mathcal{P}$ such that whenever $p(x) \in \mathcal{P}$, for all $U \to X \to (Y_1,Y_2)$ we have
$$ I(U;Y_2) \leq I(U;Y_1). $$
\end{definition}

\medskip

\begin{remark}
\label{re:rem1}
Setting $\mathcal{P}$ to be the entire set of distributions $p(x)$ shows that a less noisy receiver is in particular an essentially less noisy receiver. However,  in Section \ref{se:essln} we will show that there are essentially less noisy receivers that are not less noisy.
\end{remark}

\medskip

\begin{definition}
A receiver $Y_1$ is {\em essentially more capable} compared to receiver $Y_2$ if there exists a sufficient class of distributions $\mathcal{P}$ such that whenever $p(x) \in \mathcal{P}$, for all $U \to X \to (Y_1,Y_2)$ we have
$$ I(X;Y_2|U) \leq I(X;Y_1|U). $$
\end{definition}

\medskip

\begin{remark}
\label{re:rem2}
Clearly the above condition holds when $Y_1$ is a more capable receiver than $Y_1$, since it holds under each conditioning of $U$.
Thus by setting $\mathcal{P}$ to be entire set of distributions on $\mathcal{X}$; if $Y_1$ is also a more-capable receiver than it is also an essentially more capable receiver. 
\end{remark}
\medskip

\section{The capacity region of a broadcast channel with an essentially less noisy receiver}
\label{se:cap}

\begin{theorem}
\label{th:cap}
 The capacity region of a two-receiver broadcast channel where $Y_1$ is essentially less noisy compared to $Y_2$ is given by the union of rate pairs $(R_1,R_2)$ such that 
 \begin{align*}
   R_2 & \leq I(U;Y_2) \\
   R_1 + R_2 & \leq I(U;Y_2) + I(X;Y_1|U)
 \end{align*}
 for some $U \to X \to (Y_1,Y_2)$ and $p(x) \in \mathcal{P}$. Here $\cP$ denotes any {\em sufficient class} of distributions that makes the receiver $Y_1$ essentially less noisy compared to receiver $Y_2$.
\end{theorem}

\proof The theorem follows in a straightforward manner from the known achievability regions and outer bounds for the two receiver broadcast channels, as shown below.

\subsubsection*{The direct part}
\label{sse:ach}

It is well-known \cite{mar79} that the set of all rate pairs $(R_1,R_2)$ satisfying
\begin{align}
   R_2 & \leq I(U;Y_2) \nonumber \\
   R_1 + R_2 & \leq I(U;Y_2) + I(X;Y_1|U) \label{eq:ib}\\
   R_1 + R_2 & \leq I(X;Y_1) \nonumber
 \end{align} 
 for any $U \to X \to (Y_1,Y_2)$ is achievable via superposition coding.
Further if $p(x) \in \mathcal{P}$, we have $I(U;Y_2) \leq I(U;Y_1)$ and thus \eqref{eq:ib} reduces to the region in Theorem \ref{th:cap} and completes the proof of the achievability.

\subsubsection*{The converse part}
\label{sse:conv}

It is well-known \cite{elg79,mar79} that the set of all rate pairs $(R_1,R_2)$ satisfying
\begin{align}
   R_2 & \leq I(U;Y_2) \nonumber \\
   R_1 + R_2 & \leq I(U;Y_2) + I(X;Y_1|U) \label{eq:ob}\\
   R_1  & \leq I(X;Y_1) \nonumber
 \end{align} 
over all $U \to X \to (Y_1,Y_2)$ forms an outer bound to the capacity region of the broadcast channel. Clearly from the definition of the sufficient class $\mathcal{P}$ it is clear that one can restrict the union to be over $p(x) \in \mathcal{P}$.
Further if $p(x) \in \mathcal{P}$, we have $I(U;Y_2) \leq I(U;Y_1)$ and thus 
$$ I(U;Y_2) + I(X;Y_1|U) \leq I(X;Y_1). $$ 
This implies that \eqref{eq:ob} reduces to the region in Theorem \ref{th:cap} and completes the proof of the converse to the capacity region.

\section{A class of symmetric broadcast channels with an essentially less noisy receiver}
\label{se:essln}

In this section we prove that the class with an essentially less noisy receiver of channels is strictly larger than the class where where one receiver is {\em less noisy} \cite{kom75} compared to the other receiver. In particular this class contains the channel that served as the motivation behind this paper - the broadcast channel when one of the channels is BSC(p) and the other is BEC(e); where the pair $(p,e)$ satisfies $1 - H(p) \geq 1 - e$.

\begin{definition}
 A channel with input alphabet $\mathcal{X}$ ($\cX=\{0,1,... m-1\}$), output alphabet $\mathcal{Y}$ (of size $n$) is said to be {\em c-symmetric} if, for each $j=0,..,m-1$, there is a permutation $\pi_j(\cdot)$ of $\cY$ such that $\p(Y=\pi_j(y)|X=(i+j)_m)= \p(Y=y|X=i), \forall i,$ where $(i+j)_m = (i+j) \mod m$. 
\end{definition}

Observe that BSC and  BEC are examples of c-symmetric channels. 

A broadcast channel with input alphabet $\mathcal{X}$ and output alphabets $\mathcal{Y}_1, \mathcal{Y}_2$ is said to be {\em c-symmetric} if both the channels $ X \to Y_1$ and $X \to Y_2$ are c-symmetric.

\begin{lemma}
\label{le:unisuff}
The uniform distribution on $\mathcal{X}$ forms a sufficient class $\mathcal{P}$ for a c-symmetric broadcast channel.
\end{lemma}

\begin{proof}
Let $\mathcal{X}=\{0,1,...,m-1\}$. Given a triple $(U,V,X)$ construct a tuple $(W',U',V',X')$ as follows:
\begin{align*}
 &\p(W'=j, U'=u, V'=v, X'=i) \\
&\quad = \frac 1m \p(U=u, V=v, X=(i+j)_m).
\end{align*}
 Further set $(W',U',V') \to X' \to (Y_1', Y_2')$ to be a Markov chain with $p(y_1',y_2'|x') \equiv p(y_1,y_2|x)$ (i.e. the channel transition probability remains the same).

Observe that
\begin{align}
&\p(W'=j, U'=u, Y_1'=y) \nonumber \\
&\quad = \sum_i \p(W'=j, U'=u, X'=i, Y_1'=y) \nonumber \\
&= \frac 1m \sum_i \p(U=u, X=(i+j)_m, Y_1=\pi_j(y)) \label{eq:cst}\\
&= \frac 1m \p(U=u,Y_1=\pi_j(y)). \nonumber
\end{align}

Similarly
\begin{align}
&\p(W'=j, U'=u, Y_2'=y) 
 = \frac 1m \p(U=u,Y_2=\sigma_j(y)). \label{eq:cst1}
\end{align}

It is easy to see that the following holds:
\begin{align*}
 \p(X'=i) &= \frac 1m ~ \forall i \\
 I(X';Y_i'|W'=j) &= I(X;Y_i) ~ \forall j, i=1,2 \\
 I(U';Y_i'|W'=j) &= I(U;Y_i) ~ \forall j, i=1,2 \\
 I(X';Y_i'|U',W'=j) &= I(X;Y_i|U), ~ \forall j, i=1,2 
\end{align*}
where all equalities (except the first one) follow from equations \eqref{eq:cst}, \eqref{eq:cst1}, and that entropy is unchanged by relabeling. Similar conditions also holds for the pair $(W',V')$.

Therefore setting $\tilde{U}=(W',U'), \tilde{V}=(W',V')$ and $q(u,v,x)$ to be the distribution induced by $(\tilde{U},\tilde{V},X')$ it is easy to see that the inequalities \eqref{eq:suff} are satisfied.
As $\p(X'=i) = \frac 1m ~ \forall i$ this establishes the sufficiency of the uniform distribution.
\end{proof}

\medskip

\begin{definition}
 In a c-symmetric broadcast channel $Y_1$ is said to be  a {\em dominantly c-symmetric receiver} if the following condition holds: for every $p(x)$
  $$ I(X;Y_1)_p - I(X;Y_2)_p \leq  I(X;Y_1)_u - I(X;Y_2)_u, $$
where $u(x)$ is the uniform distribution.
\end{definition}

In other words, uniform distribution also maximizes the difference $I(X;Y_1) - I(X;Y_2)$. 

\begin{claim}
\label{cl:bsc}
For the BSC(p), BEC(e) broadcast channel $Y_1$ is a dominantly c-symmetric receiver when $1-H(p) > 1-e$. 
\end{claim}
The proof follows from part 3 of Claim \ref{cl:easy} in the appendix; also see Figure \ref{fig:btwo}.

\medskip

\begin{lemma}
\label{le:key}
In a c-symmetric broadcast channel, if $Y_1$ is a dominantly c-symmetric receiver then $Y_1$ is also an essentially less noisy receiver.
\end{lemma}

\begin{proof}
Since the uniform distribution on $\cX$ forms a sufficient class; by Lemma \ref{le:unisuff} it suffices to
show that for all $(V,X)$ such that $p(x)$ is uniform we have
\begin{equation}I(V;Y_1) \geq I(V;Y_2).\label{eq:tos} \end{equation}
Given a pair $(V,X)$ let $p_v(x)$ be the distribution on $\cX$ when $V=v$. $Y_1$ is a  dominantly c-symmetric receiver implies
$$ I(X;Y_1)_{p_v} - I(X;Y_2)_{p_v} \leq  I(X;Y_1)_u - I(X;Y_2)_u. $$

Therefore
\begin{align} 
& I(X;Y_1|V) - I(X;Y_2|V) \nonumber \\
& \quad  = \sum_v \p(V=v) \left(  I(X;Y_1)_{p_v} - I(X;Y_2)_{p_v} \right)  \nonumber \\
& \quad  \leq \sum_v \p(V=v) \left(  I(X;Y_1)_u - I(X;Y_2)_u \right)  \nonumber \\
& \quad =  I(X;Y_1)_u - I(X;Y_2)_u. \label{eq:ineq12}
\end{align}

Since $V \to X \to (Y_1,Y_2)$ is Markov and  $p(x)$ is uniform, observe
\begin{align}
& I(X;Y_1|V) - I(X;Y_2|V) \nonumber \\
&\quad = I(X;Y_1)_u - I(V;Y_1) - \left(  I(X;Y_2)_u - I(V;Y_2) \right) \nonumber \\
& \quad =  I(X;Y_1)_u - I(X;Y_2)_u - \left(  I(V;Y_1) - I(V;Y_2) \right). \label{eq:ineq13}
\end{align}

The required inequality \eqref{eq:tos} follows from \eqref{eq:ineq12} and \eqref{eq:ineq13} respectively.

\section{The capacity region of a broadcast channel with an essentially more capable receiver}
\label{se:mcap}

\begin{theorem}
\label{th:cap1}
 The capacity region of a two-receiver broadcast channel where $Y_1$ is essentially more capable compared to $Y_2$ is given by the union of rate pairs $(R_1,R_2)$ such that 
 \begin{align*}
   R_2 & \leq I(U;Y_2) \\
   R_1 + R_2 & \leq I(U;Y_2) + I(X;Y_1|U) \\
   R_1 + R_2 & \leq I(X;Y_1)
 \end{align*}
 for some $U \to X \to (Y_1,Y_2)$ and $p(x) \in \cP$.  Here $\cP$ denotes any {\em sufficient class} of distributions that makes the receiver $Y_1$ essentially more capable compared to receiver $Y_2$.
\end{theorem}

\subsubsection*{The direct part}
\label{sse:ach1}

It is well-known \cite{mar79} that the set of all rate pairs $(R_1,R_2)$ satisfying
\begin{align}
   R_2 & \leq I(U;Y_2) \nonumber \\
   R_1 + R_2 & \leq I(U;Y_2) + I(X;Y_1|U) \label{eq:ib1}\\
   R_1 + R_2 & \leq I(X;Y_1) \nonumber
 \end{align} 
 for any $U \to X \to (Y_1,Y_2)$ is achievable via superposition coding. Restricting ourselves to $p(x) \in \mathcal{P}$, we have  the region in Theorem \ref{th:cap} and completes the proof of the achievability.

\subsubsection*{The converse part}
\label{sse:conv1}

It is well-known \cite{elg79,nae07} that the set of all rate pairs $(R_1,R_2)$ satisfying
\begin{align}
   R_2 & \leq I(U;Y_2) \nonumber \\
   R_1 + R_2 & \leq I(U;Y_2) + I(X;Y_1|U) \label{eq:ob1}\\
   R_1  & \leq I(V;Y_1) \nonumber \\
   R_1 + R_2 & \leq I(V;Y_1) + I(X;Y_2|V) \nonumber
 \end{align} 
over all $(U,V) \to X \to (Y_1,Y_2)$ forms an outer bound to the capacity region of the broadcast channel. Clearly from the definition of the sufficient class $\mathcal{P}$ it is clear that one can restrict the union to be over $p(x) \in \mathcal{P}$.
Further if $p(x) \in \mathcal{P}$, since $Y_1$ is an essentially more capable receiver we have $I(X;Y_2|V) \leq I(X;Y_1|V)$ and thus 
$$ I(V;Y_1) + I(X;Y_2|V) \leq I(X;Y_1). $$ 
This implies that \eqref{eq:ob1} is contained inside the (achievable) region in Theorem \ref{th:cap} and completes the proof of the converse to the capacity region. (Indeed it is easy to see that setting $V=X$ is optimal and thus reduces the region in \eqref{eq:ob1} to the region in Theorem \ref{th:cap1}.)
\end{proof}

\section{On inclusion relationships between classes of broadcast channels}
\label{se:inc}

In this section, we present the various relationships between the classes of 2-receiver broadcast channels that were discussed in the paper.

\begin{figure}[ht]
\begin{center}
\begin{psfrags}
\psfrag{a}[c]{$I$}
\psfrag{b}[c]{$II$}
\psfrag{c}[c]{$III$}
\psfrag{d}[c]{$IV$}
\psfrag{e}[c]{$V$}
\includegraphics[width=0.7\linewidth,angle=0]{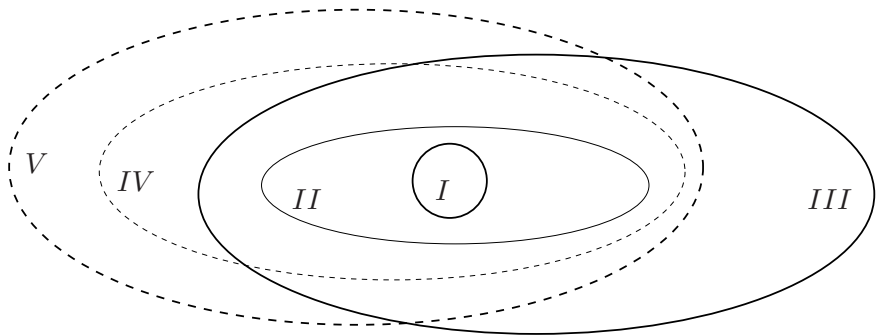}
\end{psfrags}
\caption{The classes of broadcast channels with a {\em superior} receiver. 
$I$ - degraded, $II$ - less noisy, $III$- essentially less noisy,
 $IV$ - more capable, $V$ - essentially more capable}
\label{fig:incl}
\end{center}
\end{figure}

\begin{claim}
We claim that the following relationships, as shown in Figure \ref{fig:incl}, hold
\begin{itemize}
\item[$(i)$] {\em Degraded} $\subset$ {\em less noisy} $\subset$ {\em more capable}, 
\item[$(ii)$] {\em less noisy} $\subset$ {\em essentially less noisy}, 
\item[$(iii)$] {\em essentially less noisy} $\nRightarrow$ {\em more capable}, 
\item[$(iv)$] {\em essentially less noisy} $\nRightarrow$ {\em essentially more capable}, 
\item[$(v)$] {\em more capable} $\nRightarrow$ {\em essentially less noisy}.
\item[$(vi)$] {\em more capable} $\subset$ {\em essentially more capable}, 
\end{itemize}
\end{claim}

\begin{proof}
Part $(i)$ was established in \cite{kom75}. Part $(ii)$ follows from Remark \ref{re:rem1} and Section \ref{se:essln}.  Part $(iii)$ follows from Figure \ref{fig:btwo} and Section \ref{se:essln}.

Part $(iv)$: This is again easy to deduce from Figure \ref{fig:btwo}. Take $U$ to be a binary random variable with $\p(U=0) = \frac 12$ and take $\p(X=0|U=0) = \epsilon$, $\p(X=0|U=1) = 1-\epsilon$. For sufficiently small $\epsilon$ we have that $I(X;Y_2|U) > I(X;Y_1|U)$ and hence $Y_1$ is not an {\em essentially more capable} receiver than $Y_2$. Note that $p(x)$ is uniform and  hence any sufficient class $\mathcal{P}$ must contain the uniform distribution.

Part $(v)$: Let $X \to Y_1$ be BEC(0.5) and $X \to Y_2$ be BSC(0.1101). Observe that $0.5 = 1-e > 1-H(p) \approx 0.4998$, and from part 2 of Claim \ref{cl:easy} in the Appendix we can see that $Y_1$ is a more capable receiver than $Y_2$. Let $U \to X$ be BSC(0.05), and set $\p(U=0)=0.5$. This implies $\p(X=0)=0.5 \in \mathcal{P}$ and it is easy to see that $0.3568 \approx I(U;Y_1) < I(U;Y_2) \approx 0.3924$ and thus $Y_1$ is not an essentially less noisy receiver than $Y_2$. (Note that this also implies that {\em essentially more capable} $\nRightarrow$ {\em essentially less noisy}.)

Part $(vi)$: From Remark \ref{re:rem2} it is clear that {\em more capable} $\subseteq$ {\em essentially more capable}. Hence it suffices to prove that {\em essentially more capable} $\nRightarrow$ {\em more capable}. To this end, consider the following channel.  The alphabets are given by $\mathcal{X}=\{0,1,2,3\}, \mathcal{Y}_1=\mathcal{Y}_2=\{0,1\}$.
The channel $X \to Y_1$ is a perfectly clean channel when $\mathcal{X} \in \{0,1\}$, and is the completely noisy  BSC(0.5) when $\mathcal{X}\in \{2,3\}$.  The channel $X \to Y_2$ is a BSC(0.1) when $\mathcal{X}\in \{0,1\}$ and BSC(0.4) when $\mathcal{X}\in \{2,3\}$. When $p(x)$ is uniform on $\cX = \{2,3\}$ we have $I(X;Y_2) > I(X;Y_1)$; implying $Y_1$ is not a more capable receiver than $Y_2$.
However it is easy to show that  $p(x)$  uniform on $\cX = \{0,1\}$ forms a sufficient class, and clearly on this sufficient class $Y_1$ is a more capable receiver than $Y_2$. This example shows that there are essentially more capable receivers that need not be more capable.
\end{proof}

\section*{Acknowledgement}
The author wishes to that Andrea Montanari for bringing up the motivating question  as well as some interesting discussions. The author also wishes to thank Abbas El Gamal and the anonymous referees for comments that improved the presentation.

\bibliographystyle{IEEEtran}
\bibliography{mybiblio}

\appendix
Consider a broadcast channel with two receivers. Let $X \to Y_1$ be BSC(p), $ 0 \leq p \leq \frac 12$ and $X \to Y_2$ be BEC(e).  Let
\begin{align*}
D(x)  &\stackrel{\triangle}{=} I(X;Y_1) - I(X;Y_2) \\
&= H(x*p)-(1-e)H(x)-H(p)
\end{align*} be the difference $I(X;Y_1) - I(X;Y_2)$ conditioned on  $\p(X=0) = x$. Observe that the function is symmetric about $x=\frac 12$, i.e. $D(x)=D(1-x)$.

\begin{claim}
\label{cl:easy}
The function $D(x)$ has the following properties:
\begin{enumerate}
\item When $e \leq 2p$, $D(x)$ monotonically decreases in the interval $[0,\frac 12]$.
\item When $2p < e \leq H(p)$, $D(x)$ monotonically decreases in the interval $[0,r]$, and
monotonically increases in the interval $[r, \frac 12]$ for some $r \in (0, \frac 12]$. The maximum occurs at $x=0$, i.e. $D(x) \leq 0, \forall x \in [0,1]$.
\item When $H(p) < e \leq 1$, $D(x)$ monotonically decreases in the interval $[0,r]$, and
monotonically increases in the interval $[r, \frac 12]$ for some $r \in (0, \frac 12)$. The maximum occurs at $x=\frac 12$, i.e. $D(x) \leq D(\frac 12), \forall x \in [0,1]$.
\end{enumerate}
\end{claim}
\begin{proof}
Let $J(x) = \log_2\frac{1-x}{x}$. Observe that
\begin{equation}\frac{d}{dx} D(x) = (1-2p)J(x*p) - (1-e)J(x).\label{eq:eq123} \end{equation}
For $x \in [0, \frac 12]$, when $e \leq 2p$ we have  $\frac{d}{dx} D(x) \leq 0$, since  $0<J(x*p) < J(x)$  and establishes Part 1. 

From \eqref{eq:eq123} any $x$ such that $\frac{d}{dx} D(x) = 0$ must satisfy
\begin{equation} \label{eq:zero} \left( \big(\frac{1-x}{x}\big)^c + 1\right)^{-1} = x(1-p)+p(1-x),\end{equation}
where $c=\frac{1-e}{1-2p}$. Define
$$L(x) =  \left( \big(\frac{1-x}{x}\big)^c + 1\right)^{-1}. $$
When $e \geq 2p$, we have $0 < c < 1$. Then it is easy to see that $L(x)$ is concave in $x \in [0,\frac 12]$. Observe that
$$\frac{d}{dx} L(x) = \frac{c}{\big(x(1-x)\big)^{1-c}\big[ (1-x)^c + x^c)\big]^2 x^{1-c}} $$
and since the functions: $\big(x(1-x)\big)^c, (1-x)^c + x^c, x^{1-c}$ increase in $x \in [0,\frac 12]$, we have $\frac{d^2}{dx^2}L(x) \leq 0$. This implies that $L(x)$ can intersect the line $R(x)=x(1-p)+p(1-x)$ at possibly no more than two points on $x \in [0,\frac 12]$. Since $L(\frac 12) = \frac 12 = R(\frac 12)$, there is  at most one other  solution $r \in (0,\frac 12)$ to \eqref{eq:zero} when $x \in (0,\frac 12)$. 

Since $\frac{d}{dx} D(x) \to -\infty$ as $x \to 0^+$, it is clear that $D(x)$ decreases in $[0,r]$ and increases in $[r,\frac 12]$. The maximum can therefore be obtained by comparing $D(0)=0$ and $D(\frac 12)=e-H(p)$. This establishes Parts 2,3. 
\end{proof}

\begin{claim}
Consider a broadcast channel with two receivers. Let $X \to Y_1$ be BSC(p), $ 0 \leq p \leq \frac 12$ and $X \to Y_2$ be BEC(e). Then the following holds:
\begin{enumerate}
\item $0 \leq e \leq 2p$: $Y_1$ is a degraded version of $Y_2$.
\item $2p \leq e \leq 4p(1-p)$: $Y_2$ is less noisy  than  $Y_1$, but is not a degraded version.
\item $4p(1-p) \leq e \leq H(p)$: $Y_2$ is more capable (but not less noisy) than $Y_2$.
\item $H(p) \leq e \leq 1$: $Y_1$ is essentially less noisy than  $Y_2$.
\end{enumerate}
\end{claim}
\begin{proof}
Part 1 is well-known and easy to establish. Part 4 follows from Claim \ref{cl:easy} and Section \ref{se:essln}. We also know from Parts 1 and 2 of Claim \ref{cl:easy} that when $0 \leq e \leq H(p)$ $D(x) \leq 0$, i.e. $Y_2$ is a more capable receiver than $Y_1$. Therefore to complete the proof of the claim it suffices to show that $Y_2$ is a less noisy receiver than $Y_1$ {\em if and only if} $0 \leq e \leq 4p(1-p)$.

\noindent The following statements are equivalent (the proof is immediate and omitted): 

$(i)$ $Y_2$ is a less noisy receiver than $Y_1$;

$(ii)$ $\forall U \to X \to (Y_1,Y_2)$, $I(X;Y_1|U) - I(X;Y_2|U) \geq I(X;Y_1) - I(X;Y_2)$;

$(iii)$ $I(X;Y_1)_{p(x)} - I(X;Y_2)_{p(x)}$ is a convex function of $p(x)$.

Therefore $Y_2$ is a less noisy receiver than $Y_1$ {\em if and only if} $D(x)$ is convex for $x \in [0,1]$. It is straight forward to see that $D(x)$ is convex {\em if and only if} $0 \leq e \leq 4p(1-p)$.
\end{proof}

\end{document}